\RequirePackage{fix-cm}
\documentclass{svjour3}                     
\smartqed  
\usepackage{graphicx}

\usepackage{paralist}
\usepackage{amsmath}
\usepackage{amssymb}
\usepackage{lipsum}
\usepackage{hyperref}
\usepackage[super]{nth}

\begin{document}

\title{On Hedonic Games with Common Ranking Property\footnote{A preliminary version of this work have appeared in the proceedings of the \nth{11} International Conference on Algorithms and Complexity (CIAC 2019) \cite{ciac19}}
}


\author{Bugra Caskurlu \and
        Fatih Erdem Kizilkaya 
}


\institute{Bugra Caskurlu \and Fatih Erdem Kizilkaya \at
        TOBB University of Economics and Technology, \\
            Department of Computer Engineering, \\
            Ankara, Turkey        \\
            \email{\{bcaskurlu,fkizilkaya\}@etu.edu.tr}\\
}

\date{Received: date / Accepted: date}

\maketitle

\begin{abstract}
Hedonic games are a prominent model of coalition formation, in which each agent's utility only depends on the coalition she resides. The subclass of hedonic games that models the formation of general partnerships \cite{law}, where output is shared equally among affiliates, is referred to as hedonic games with common ranking property (HGCRP). Aside from their economic motivation, HGCRP came into prominence since they are guaranteed to have core stable solutions that can be found efficiently \cite{partnerships}. We improve upon existing results by proving that every instance of HGCRP has a solution that is Pareto optimal, core stable and individually stable. The economic significance of this result is that efficiency is not to be totally sacrificed for the sake of stability in HGCRP. We establish that finding such a solution is {\bf NP-hard} even if the sizes of the coalitions are bounded above by $3$; however, it is polynomial time solvable if the sizes of the coalitions are bounded above by $2$. We show that the gap between the total utility of a core stable solution and that of the socially-optimal solution (OPT) is bounded above by $n$, where $n$ is the number of agents, and that this bound is tight. Our investigations reveal that computing OPT is inapproximable within better than $O(n^{1-\epsilon})$ for any fixed $\epsilon > 0$, and that this inapproximability lower bound is polynomially tight. However, OPT can be computed in polynomial time if the sizes of the coalitions are bounded above by $2$.
\end{abstract}

\section{Introduction}
\label{sec:intro}
The class of games, where a finite set of agents are to be partitioned into groups (coalitions) is known as \textit{coalition formation games} \cite{coal}. A coalition formation game is said to be a \textit{hedonic coalition formation game} (or simply \textit{hedonic game}) if the utilities of the agents exclusively depend on the coalition they belong to \cite{seminal-hedonic}, i.e., the agents do not worry about how the remaining agents are partitioned.

Hedonic games subsume a wide range of problems, which includes well-known matching problems, such as the stable marriage \cite{marriage}, the stable roommates \cite{roommates}, and the hospital/residents problems \cite{marriage}. Moreover, hedonic games also found their way into artificial intelligence research since they provide a simple yet versatile formal model for various coordination problems in distributed and multi-agent systems \cite{multiagent-systems}. Recently, hedonic games have been used to develop a novel framework to solve a task allocation problem in swarm robotics \cite{swarm-robotics}.

In this paper, we focus on another important application of hedonic games referred to as \textit{formation of partnerships}, which corresponds to one of the first and the most natural subclasses of hedonic games studied by economists in the late 1980s. However, computational aspects of this subclass of hedonic games remained largely\footnote{To our knowledge, the only known result about this subclass of hedonic games is that there always exists a core stable partition which can be found in polynomial time.} open since then.

Partnerships are formal agreements made by the founders of for-profit start-up ventures when the company is founded \cite{law}. Even though there are several different types of partnerships, the most common form is referred to as general partnership, where all parties share the legal and financial liabilities as well as the profit of the partnership equally. A hedonic game, where all agents of a coalition receive the same utility (as is the case in general partnerships) is said to possess \textit{common ranking property} \cite{partnerships}. Hedonic games with common ranking property (HGCRP) can, of course, model the formation of groups other than partnerships. For example, the formation of class project groups where the students in the same group receive the same grade can be modeled in that setting \cite{sagt21}.

In game theory, the main objective is to determine the viable outcomes of a game under the assumption that the agents are rational. This is typically done by means of solution concepts. Each solution concept, designed with the intent of modeling the behavior of rational agents, comes with a set of axioms to be satisfied by the outcome of the game. The outcome of a hedonic game is a partition of agents, which is also referred to as a \textit{coalition structure}.

Existence and tractability of Nash stable \cite{nash-stable}, individually stable \cite{individually}, contractually individually stable \cite{contractually} and core stable \cite{core-stable} outcomes of various subclasses of hedonic games are extensively studied. The most prominent solution concept used in hedonic games literature is the core stability, where a partition of agents is defined as stable if no subset of agents can form a new coalition together so that the utilities of all the agents in this subset are increased \cite{core-stable}.

Notice that core stability is the stability concept used in the definition of classical matching problems such as the stable marriage and the stable roommates. Some subclasses of hedonic games, such as the stable roommates, may not have any core stable outcome. However, every HGCRP instance is guaranteed to have a core stable outcome and such an outcome can be computed via a simple greedy algorithm \cite{partnerships}.

Since forming a completely new coalition might be hard to coordinate, there are also stability concepts resilient to single-agent deviations. An outcome is Nash stable if no agent can benefit by moving from her coalition to another existing (possibly empty) coalition. An outcome is said to be individually stable if no agent can benefit by moving from her coalition to another existing (possibly empty) coalition while not making the members of her new coalition worse off. Notice that a Nash stable outcome is also individually stable, but not vice versa. A core stable outcome is not necessarily individually stable, and an individually stable outcome is not necessarily core stable.

Pareto optimality is the most widely adopted concept of efficiency in game theory \cite{pareto}. In hedonic games, a partition $\pi$ of agents is called a Pareto optimal coalition structure, if no other partition $\pi'$ of agents makes some agents better off without making some agents worse off. Pareto optimality is a desirable property for the outcome of a game, and every finite game is guaranteed to have a Pareto optimal outcome. However, note that Pareto optimal outcomes do not necessarily coincide with stable outcomes of a game. As an extreme example, in the classical Prisoner's Dilemma game \cite{agt}, all unstable outcomes are Pareto optimal, whereas the unique stable outcome is not.

\medskip

\textbf{Our Contributions}

\medskip

A core stable coalition structure of an HGCRP instance is not necessarily Pareto optimal (see Example \ref{Ex1}), and a Pareto optimal coalition structure of an HGCRP instance is not necessarily core stable (see Example \ref{Ex2}). It is also clear that individually stable outcomes do not coincide with Pareto optimal or core stable outcomes. Thus, a natural research question (RQ) is as follows:

\medskip

\textbf{RQ1.} \textit{Does there always exist a partition $\pi$ of agents of a given HGCRP instance, such that $\pi$ is Pareto optimal, core stable and individually stable?}

\medskip

Theorem \ref{thm:Existence} answers this question affirmatively via a constructive proof using the potential function method \cite{potential}.

\medskip

Since we resolved the existence problem, the next immediate research direction is establishing the computational complexity of finding such a partition of agents of a given HGCRP instance. In order to address this problem, we first need to state how HGCRP instances are represented.

Since the representation of a hedonic game instance requires exponential\footnote{In a hedonic game, the utility value of each agent is defined over all subsets of agents containing her. Thus, if $n$ is the number of agents, the input size is $O(n \cdot 2^{n - 1})$. For HGCRP instances, we do not need to define separate utility values for each agent, rather, we only need to have a joint utility value for each nonempty coalition (since all agents in the coalition have the same utility), and thus the input size is $O(2^n - 1)$.
} space in the number of agents, a lot of effort is spent in defining subclasses of hedonic games with concise representations\footnote{A representation of a hedonic game is generally called concise if it is polynomial in the number of agents (for example, see \cite{nets}). However, note that such a representation would not be capable of representing all HGCRP instances.}. In this paper, we assume that an HGCRP instance is represented as \textit{individually rational coalition lists} \cite{ballester}, i.e., the utility value of a subset $S$ of agents is omitted, if the utility value of $S$ is smaller than that of one of its singleton subsets. This is because no stable outcome can contain such a coalition. Since we have a model for input representation, we are ready to state the next immediate RQ as follows.

\medskip

\textbf{RQ2.} \textit{What is the computational complexity of finding a Pareto optimal, core stable and individually stable partition of a given HGCRP instance?}

\medskip

The problem is {\bf NP-hard} as stated by Corollary \ref{cor:Complexity}, since finding any Pareto optimal partition is {\bf NP-hard} due to Theorem \ref{thm:Complexity}. On the other hand, finding a core stable and individually stable partition is polynomial time solvable due to Theorem \ref{thm:GreedyAlg}.

\medskip

An outcome of the game maximizing the sum of the utilities of the agents is referred to as a \textit{socially-optimal solution} (OPT). Since a socially-optimal solution is Pareto optimal by definition, finding a socially-optimal solution is also {\bf NP-hard}, which brings us to the next immediate RQ.

\medskip

\textbf{RQ3.} \textit{How closely is it possible to approximate a socially-optimal solution of a given HGCRP instance?}

\medskip

Computing OPT of a given HGCRP instance is inapproximable within better than $O(n^{1-\epsilon})$ for any fixed $\epsilon > 0$ (unless {\bf P = NP}) due to Theorem \ref{thm:Inapproximability}, where $n$ is the number of agents. Moreover, this lower bound is polynomially tight since computing OPT is $n$-approximable due to Corollary \ref{cor:Approximability}.

\medskip

There are several metrics in the literature to quantify the loss of total utility in a game due to the selfish behavior of the agents. The most popular of these metrics are the \textit{price of anarchy} \cite{poa}, and the \textit{price of stability} \cite{pos}. The price of anarchy and the price of stability of a game are defined as the supremum of the ratio of the total utility of OPT to that of the socially worst and best stable solutions, respectively, over all instances of the game. In this paper, by a stable solution, we mean a core stable partition of agents. The next immediate RQ is about these metrics.

\medskip

\textbf{RQ4.} \textit{What are the price of anarchy and the price of stability of HGCRP?}

\medskip

Both the price of anarchy and the price of stability of HGCRP are $n$, where $n$ is the number of agents, due to Theorem \ref{thm:Quality}.

\medskip

Since almost all of the computational problems we study in this paper turn out to be intractable, it is of great interest to explore restrictions where those problems might be tractable. The most natural restriction is bounding the sizes of the coalitions. We consider this restriction in the last RQ.

\medskip

\textbf{RQ5.} \textit{ Is any of the above computational problems tractable in the restricted setting where the sizes of the coalitions are bounded by a constant?}

\medskip

When the sizes of the coalitions are bounded above by $2$, not only the problem of finding a socially-optimal solution, but also the problem of finding a Pareto optimal, core stable, and individually stable partition can be solved in polynomial time due to Theorems \ref{thm:MathchingAlg1} and \ref{thm:MathchingAlg2}, respectively. However, even when the sizes of the coalitions are bounded above by $3$, finding a Pareto optimal partition is still {\bf NP-hard} by Corollary \ref{cor:FixedPerfect}; and finding a socially-optimal solution is {\bf APX-complete} by Corollary \ref{cor:Inapproximability}.

\section{Notation and Preliminaries}
\label{sec:notation}

We define an instance of HGCRP as a binary pair $\mathcal{G} = (N, U)$, where $N$ is a finite set of $n$ agents, and $U : 2^N \setminus \emptyset \rightarrow \mathbb{R_+}$ is a non-negative real-valued function defined over the nonempty subsets of $N$. We assume that an instance is represented as \textit{individually rational coalition lists} (IRCL) \cite{ballester}, i.e., if $U(S) < U(\{i\})$ for some agent $i \in S$ for a subset of agents $S \subseteq N$, then $U(S)$ is omitted from the list, since $S$ is not a viable coalition in any stable coalition structure.

\medskip

The solution (outcome) of a game is a partition (coalition structure) $\pi$ over the set of agents $N$. The coalition containing an agent $i \in N$ in partition $\pi$ is denoted by $\pi(i)$. In a partition $\pi$, the utilities of all the agents in the same coalition $S \in \pi$ are the same, and equal to the \textit{joint utility} $U(S)$. We use $u_i(\pi)$ to denote the \textit{utility} of some agent $i$ in partition $\pi$. Notice that $u_i(\pi) = U(\pi(i))$.

\medskip

A nonempty subset $S \subseteq N$ of agents is said to be a \textit{blocking coalition with respect to partition} $\pi$,
if $U(S) > u_i(\pi)$ for all agents $i \in S$, i.e., any agent $i \in S$ is \textit{strictly} better off in $S$ than she is in $\pi(i)$. A coalition structure $\pi$ is \textit{core stable} if there is no blocking coalition with respect to it.

\medskip

For a partition $\pi$ and a coalition $S$, we define $\pi_S$ as the partition induced on $\pi$ by $S$. $\pi_S$ is the partition that would arise if the agents in $S$ collectively deviated from $\pi$ to form coalition $S$, i.e., $\pi_S(i) = S$ for all $i \in S$, and $\pi_S(j) = \pi(j) \setminus S$ for all $j \in N \setminus S$. Notice that $\pi_S$ may or may not be core stable.

\medskip

Pareto optimality is a measure to assess the social quality of the solutions of a game with respect to alternative solutions of the same game. For two coalition structures $\pi$ and $\pi'$ over the set of agents $N$, we say that $\pi'$ \textit{Pareto dominates} $\pi$ if, $u_i(\pi') \ge u_i(\pi)$ for all agents $i \in N$, and there exists an agent $i$ for which the inequality is strict. In other words, if $\pi'$ Pareto dominates $\pi$, then all agents are at least as good in $\pi'$ than in $\pi$, and there is an agent that is strictly better off in $\pi'$ than in $\pi$. A coalition structure $\pi$ is said to be \textit{Pareto optimal} if no coalition structure Pareto dominates it.

\bigskip

We next present two examples to illustrate the notions of core stability and Pareto optimality. In Example \ref{Ex1}, we present a core stable partition that is not Pareto optimal, and in Example \ref{Ex2}, we present a Pareto optimal partition that is not core stable. That is, Pareto optimality and core stability are independent concepts.

\begin{example}
\label{Ex1}
Let $\mathcal{G} = (N, U)$ be an HGCRP instance, where $N = \{1, 2\}$, and $U$ is defined as $U(\{1\}) = U(\{1, 2\}) = 1$, and $U(\{2\}) = 0$. Notice that the coalition structure $\pi = \{\{1\}, \{2\}\}$ is core stable since there is no blocking deviation with respect to it. But $\pi$ is not Pareto optimal since the partition $\pi' = \{\{1,2\}\}$ Pareto dominates $\pi$.
\end{example}

\medskip

\begin{example}
\label{Ex2}
Let $\mathcal{G} = (N, U)$ be an HGCRP instance, where $N = \{1, 2, 3\}$, and $U$ is defined as $U(\{1\}) = 0$, $U(\{2, 3\}) = 2$ and $U(S) = 1$ for all other subsets $S \subseteq N$ of agents. Notice that partition $\pi = \{\{1, 2\}, \{3\}\}$ is Pareto optimal since no partition Pareto dominates it. However, $\pi$ is not core stable since $S = \{2, 3\}$ is a blocking coalition with respect to $\pi$. Notice that $\pi_S = \{\{1\}, \{2, 3\}\}$.
\end{example}

The most prominent stability concepts using only unilateral deviations in the hedonic games literature are as follows:

\medskip

\begin{compactitem}
\item[--] A partition $\pi$ is \textit{Nash stable} if no agent can benefit by moving from her coalition $S \in \pi$ to an existing (possibly empty) coalition $T \in \pi \cup \{\emptyset\}$.
\item[--] A partition $\pi$ is \textit{individually stable} if no agent can benefit by moving from her coalition $S \in \pi$ to an existing (possibly empty) coalition $T \in \pi \cup \{\emptyset\}$ while not making the members of $T$ worse off.
\end{compactitem}

\medskip

Notice that a Nash stable partition is also individually stable by definition, but not vice versa. An HGCRP instance does not necessarily have a Nash stable partition as can be seen in Example \ref{Ex3}. Therefore, we rather focus on individual stability in this paper.

\begin{example}
\label{Ex3}
Let $\mathcal{G} = (N, U)$ be an HGCRP instance, where $N = \{1, 2\}$, and $U$ is defined as $U(\{1\}) = 1$, $U(\{1, 2\}) = 2$, and $U(\{2\}) = 3$. Notice that this HGCRP instance is simply the stalker game. Therefore, there does not exist a Nash stable partition. On the other hand, $\pi = \{\{1\}, \{2\}\}$ is individually stable.
\end{example}



\section{Pareto Optimal, Core Stable and Individually Stable Partitions}
\label{sec:ParOptCorStaPar}

In this section, we first prove the existence of a Pareto optimal, core stable and individually stable partition for any HGCRP instance. We then establish the computational complexity of finding such a partition.

We prove the existence result, given below as Theorem \ref{thm:Existence}, by presenting a potential function defined over the set of coalition structures of the given HGCRP instance. The potential function is maximized at a Pareto optimal, core stable, and individually stable partition.

\begin{theorem}
\label{thm:Existence}
Every HGCRP instance $\mathcal{G} = (N, U)$ has a coalition structure that is Pareto optimal, core stable and individually stable at the same time.
\end{theorem}

\begin{proof}
We fix an HGCRP instance $(N, U)$ with $n$ agents. We define $\psi(\pi)$ as the sequence of the utilities of the agents in partition $\pi$ in \textit{non-increasing order}. We use the symbols $\rhd$ and $\unrhd$, respectively, to denote the binary relations ``lexicographically greater than'' and ``lexicographically greater than or equal to'' over the set of sequences of utilities of agents. We also denote the $i^{th}$ element in the sequence $\psi(\pi)$ by $\psi_i(\pi)$.

\begin{lemma}
\label{lem:Pareto}
If $\pi'$ Pareto dominates $\pi$, then $\psi(\pi') \rhd \psi(\pi)$.
\end{lemma}

\begin{proof}
We first rename the agents so that $\psi_i(\pi)$ is the utility of agent $i$ in partition $\pi$. Notice that we have $u_1(\pi) \geq u_2(\pi) \geq \ldots \geq u_n(\pi)$. Let $\psi'(\pi')$ be the permutation of $\psi(\pi')$, where $\psi'_i(\pi')$ and $\psi_i(\pi)$ are the respective utilities of the same agent for all $i \in N$. Notice that $\psi'_i(\pi') = u_i(\pi')$ and $\psi_i(\pi) = u_i(\pi)$.

Since $\pi'$ Pareto dominates $\pi$, we have $u_i(\pi') \geq u_i(\pi)$ for all $i \in N$, and $u_j(\pi') > u_j(\pi)$ for some agent $j$. Hence, $\psi'_i(\pi') \geq \psi_i(\pi)$ for all $i \in N$, and $\psi'_j(\pi') > \psi_j(\pi)$ for some agent $j$. Therefore, $\psi'(\pi') \rhd \psi(\pi)$. Notice that $\psi(\pi') \unrhd \psi'(\pi')$ since $\psi(\pi')$ is the same sequence as $\psi'(\pi')$ but sorted in non-increasing order. Hence, $\psi(\pi') \unrhd \psi'(\pi') \rhd \psi(\pi)$. \qed
\end{proof}

\begin{lemma}
\label{lem:Core}
If $S$ is a blocking coalition with respect to $\pi$, then $\psi(\pi_S) \rhd \psi(\pi)$.
\end{lemma}

\begin{proof}
Due to the common ranking property, we can assume without loss of generality, that the utilities of agents that are in the same coalition in partition $\pi$ are listed consecutively in $\psi(\pi)$. Moreover, we can assume without loss of generality, that for a coalition $G$ in partition $\pi$, the utilities of agents in $G \cap S$ precede in the ordering of $\psi(\pi)$ of those in $G \setminus S$. We also rename agents such that $\psi_i(\pi)$ correspond to the utility of agent $i$ in partition $\pi$. Notice that we have $u_1(\pi) \geq u_2(\pi) \geq \ldots \geq u_n(\pi)$.

Let $\psi'(\pi_S)$ be a permutation of $\psi(\pi_S)$, where $\psi'_i(\pi_S)$ and $\psi_i(\pi)$ are the respective utilities of the same agent for all $i \in N$. Then, $\psi'_i(\pi_S) = u_i(\pi_S)$ and $\psi_i(\pi) = u_i(\pi)$.

Let $i$ be the agent with the smallest index such that $\psi_i(\pi) >  \psi'_i(\pi_S)$. That is, $u_i(\pi) > u_i(\pi_S)$.  This implies $i \notin S$, since otherwise we would have  $u_i(\pi) < u_i(\pi_S)$. Additionally, there must be an agent $j \in \pi(i)$ such that $j \in S$, since otherwise we would have $\pi(i) = \pi_S(i)$, and it would mean $u_i(\pi) = u_i(\pi_S)$. Note that $u_j(\pi_S) > u_j(\pi)$ since $j \in S$, which means $\psi'_j(\pi_S) > \psi_j(\pi)$. Also recall that $u_j(\pi)$ precedes $u_i(\pi)$ in the ordering of $\psi(\pi)$. Therefore, $\psi'(\pi_S) \rhd \psi(\pi)$ since $i$ is the agent with the smallest index such that $\psi_i(\pi) > \psi'_i(\pi_S)$, and there exists an agent $j < i$ such that $\psi_j(\pi) <  \psi'_j(\pi_S)$. Notice that $\psi(\pi_S) \unrhd \psi'(\pi_S)$ since $\psi(\pi_S)$ is the same sequence as $\psi'(\pi_S)$ but sorted in non-increasing order. Therefore, $\psi(\pi_S) \unrhd \psi'(\pi_S) \rhd \psi(\pi)$. \qed
\end{proof}

\begin{lemma}
\label{lem:IS}
If $S \in \pi$ is a coalition such that $U(S \cup \{i\}) > U(\pi(i))$ and $U(S \cup \{i\}) \geq U(S)$ for some agent $i \notin S$, then $\psi(\pi_{S \cup \{i\}}) \rhd \psi(\pi)$.
\end{lemma}

\begin{proof}
Partitions $\pi_{S \cup \{i\}}$ and $\pi$ differ from each other only in one respect: partition $\pi_{S \cup \{i\}}$ contains coalitions $S \cup \{i\}$ and $\pi(i) \setminus \{i\}$ whereas partition $\pi$ contains coalitions $S$ and $\pi(i)$ instead. Since that $U(S \cup \{i\}) > U(\pi(i))$ and $U(S \cup \{i\}) \geq U(S)$ by definition, we have $\psi(\pi_{S \cup \{i\}}) \rhd \psi(\pi)$.
\end{proof}

Let $\pi^*$ be a coalition structure such that $\psi(\pi^*) \unrhd \psi(\pi)$ for all partitions $\pi$ over the set of agents $N$. Note that such a partition $\pi^*$ exists, since the set of partitions over $N$ is finite.

\bigskip

\begin{compactitem}
\item[--] $\pi^*$ is Pareto optimal, since otherwise there is a partition $\pi$ such that $\pi$ Pareto dominates $\pi^*$. But then $\psi(\pi) \rhd \psi(\pi^*)$ by Lemma \ref{lem:Pareto}, which contradicts the fact that $\psi(\pi^*) \unrhd \psi(\pi)$ for all partitions $\pi$ over $N$.

\item[--] $\pi^*$ is core stable, since otherwise there is a blocking coalition $S$ with respect to $\pi^*$. But then $\psi(\pi^*_S) \rhd \psi(\pi^*)$ by Lemma \ref{lem:Core}, which contradicts the fact that $\psi(\pi^*) \unrhd \psi(\pi)$ for all partitions $\pi$ over $N$.

\item[--] $\pi^*$ is individually stable, since otherwise there is a coalition $S \in \pi^*$ such that $U(S \cup \{i\}) > U(\pi^*(i))$ and $U(S \cup \{i\}) \geq U(S)$ for some agent $i \notin S$. But then $\psi(\pi^*_{S \cup \{i\}}) \rhd \psi(\pi^*)$ by Lemma \ref{lem:IS}, which contradicts the fact that $\psi(\pi^*) \unrhd \psi(\pi)$ for all partitions $\pi$ over $N$.
\end{compactitem}

\bigskip

This means that every HGCRP instance has a coalition structure that is Pareto optimal, core stable and individually stable. \qed
\end{proof}

The rest of the section is devoted to proving that finding a Pareto optimal, core stable and individually stable partition of a given HGCRP instance is {\bf NP-hard}. To do that, we first establish the intractability of finding any Pareto optimal partition of a given HGCRP instance, which trivially reduces to finding a Pareto optimal, core stable and individually stable partition.

In our proof, we make use of an interesting result in the literature \cite{pareto} that relates the computational complexity of finding a Pareto optimal partition of a subclass of hedonic games to that of finding a perfect partition in the same. A partition of agents in a hedonic game is said to be a \textit{perfect partition}, if every agent is in her most preferred coalition, i.e., receives the maximum utility she can attain. Notice that a hedonic game (and also an instance of HGCRP) does not necessarily possess a perfect partition. Hence, the problem of finding a perfect partition is formally specified as follows.

\medskip

PERFECT-PARTITION = \textit{``Given a hedonic game instance, return a perfect partition if exists, return $\emptyset$ otherwise.''}

\medskip

The aforementioned relation between the computational complexities of finding a Pareto optimal partition and PERFECT-PARTITION is established via the following result in \cite{pareto}.

\begin{theorem}[Aziz et al. \cite{pareto}] \label{thm:Aziz}
For every class of hedonic games, where it can be checked whether a given partition is perfect in polynomial time, {\bf NP-hardness} of PERFECT-PARTITION implies {\bf NP-hardness} of computing a Pareto optimal partition.
\end{theorem}

Since it can be efficiently checked whether a given partition $\pi$ of an HGCRP instance in IRCL representation is perfect, Theorem \ref{thm:Aziz} implies that all we need to complete the proof is to show that PERFECT-PARTITION is {\bf NP-hard} for the subclass HGCRP, which is stated as Theorem \ref{thm:Perfect}.

\begin{theorem} \label{thm:Perfect}
PERFECT-PARTITION is {\bf NP-hard} for HGCRP.
\end{theorem}

\begin{proof}
Since the search version of an {\bf NP-complete} decision problem is {\bf NP-hard} \cite{search}, all we need is to show that deciding existence of a perfect partition of an HGCRP instance is {\bf NP-complete}. We will do that by giving a polynomial time mapping reduction from the EXACT-COVER problem \cite{karp}. In the EXACT-COVER problem, we are given a universe $\mathcal{U} = \{1, \ldots, n\}$ and a family $\mathcal{S} = \{\mathcal{S}_1, \ldots, \mathcal{S}_k\}$ of subsets of $\mathcal{U}$. We are asked to decide whether there exists an exact cover $\mathcal{C} \subseteq \mathcal{S}$, i.e., each element in $\mathcal{U}$ is contained in exactly one subset in $\mathcal{C}$.

\bigskip

Given an instance $(\mathcal{U,S})$ of EXACT-COVER, we construct a corresponding instance $(N,U)$ of HGCRP as follows:

\medskip

\begin{compactitem}
\item[-] For every element $i \in \mathcal{U}$, there is a corresponding agent $i \in N$.
\item[-] Each subset $\mathcal{S}_j \in \mathcal{S}$ corresponds to a subset of agents $S_j \subseteq N$ with joint utility $2$, i.e., $U(S_j) = 2$.
\item[-] For every subset $S \subseteq N$ of agents for which there is no corresponding subset in $\mathcal{S}$, $U(S) = 1$ if $\lvert S \rvert = 1$, and $U(S) = 0$ otherwise.
\end{compactitem}

\medskip

Notice that for the constructed HGCRP instance $(N,U)$ in IRCL representation, the joint utilities are only given for the singleton coalitions, and the coalitions with a corresponding subset in $\mathcal{S}$. Hence, the size of $(N,U)$ is polynomial in the size of the given EXACT-COVER instance $(\mathcal{U,S})$. We next show that $(\mathcal{U,S})$ has an exact cover, if and only if there exists a perfect partition in $(N,U)$.

\bigskip

\textit{(Only If)} Suppose that $(\mathcal{U,S})$ a exact cover $\mathcal{C}$. Let $\pi$ be the coalition structure of $(N,U)$, where each coalition $S_j \in \pi$ is the correspondent of a subset $\mathcal{S}_j \in \mathcal{C}$. Notice that $\pi$ is a partition over the set of agents of $(N,U)$ since $\mathcal{C}$ is an exact cover. Since $u_i(\pi) = 2$ for all agents $i \in N$, all agents are in their most preferred coalition, and hence, $\pi$ is a perfect partition of $(N,U)$.

\medskip

\textit{(If)} Suppose that there exists a perfect partition $\pi$ of $(N,U)$, i.e., $u_i(\pi) = 2$ for every agent $i \in N$. Then, for each coalition $S_i \in \pi$, there is a corresponding subset $\mathcal{S}_i \in \mathcal{S}$. Notice that these subsets not only cover all elements of $\mathcal{U}$ but also do not overlap since $\pi$ is a partition, and hence, form an exact cover.

\bigskip

Since it is trivial to check efficiently whether a partition of a given HGCRP instance in IRLC representation is a perfect partition, the decision version of PERFECT-PARTITION is in {\bf NP}, and thus {\bf NP-complete} by the above reduction. \qed
\end{proof}

Theorem \ref{thm:Complexity} is a direct consequence of Theorem \ref{thm:Aziz} and Theorem \ref{thm:Perfect}.

\begin{theorem} \label{thm:Complexity}
Finding a Pareto optimal partition of a given HGCRP instance is {\bf NP-hard}.
\end{theorem}

Even though a core stable partition can be computed in polynomial time by a simple greedy algorithm \cite{partnerships}, Theorem \ref{thm:Complexity} implies that we cannot find a Pareto optimal, core stable and individually stable partition of a given HGCRP instance in polynomial time, as stated by Corollary \ref{cor:Complexity}.

\begin{corollary} \label{cor:Complexity}
Finding a Pareto optimal, core stable and individually stable partition of a given HGCRP instance is {\bf NP-hard}.
\end{corollary}
Since the EXACT-COVER is {\bf NP-complete}, even under the restriction that $\lvert \mathcal{S}_{i} \rvert = 3$ for all $\mathcal{S}_i \in \mathcal{S}$ \cite{3-set}, the mapping reduction used in the proof of Theorem \ref{thm:Perfect} also establishes that PERFECT-PARTITION is {\bf NP-hard} for HGCRP, even under the restriction that the sizes of the coalitions are bounded above by $3$.

As a consequence of the reduction given in \cite{pareto} from finding a perfect partition to finding a Pareto optimal partition, we obtain the following result.

\begin{corollary} \label{cor:FixedPerfect}
All of the following problems are {\bf NP-hard} for HGCRP, even under the restriction that the sizes of the coalitions are bounded above by $3$.

\begin{itemize}
  \item Finding a perfect partition,
  \item Finding a Pareto optimal partition,
  \item Finding a Pareto optimal, core stable and individually stable partition.
\end{itemize}

\end{corollary}

We next show that a core stable and individually stable partition of a given HGCRP instance can be computed in polynomial time by a simple greedy algorithm.

\begin{theorem} \label{thm:GreedyAlg}
For any given HGCRP instance, a core stable and individually stable partition can be computed in polynomial time.
\end{theorem}

\begin{proof}
For a given HGCRP instance $(N, U)$ in IRCL representation, let $\mathtt{IRCL}$ denote the individually rational coalition list of $(N, U)$. From $\mathtt{IRCL}$, pick a coalition $\mathtt{C}$ with maximum joint utility. If there are multiple such coalitions in $\mathtt{IRCL}$, let $\mathtt{C})$ be a coalition with the largest cardinality among these coalitions. Since $U(\mathtt{C}) \geq U(S)$ for any subset of agents $S$, the members of $\mathtt{C}$ cannot participate in a blocking coalition with respect to any partition containing $\mathtt{C}$. Furthermore, if any other agent moves to coalition $\mathtt{C}$, then the members of $\mathtt{C}$ get worse off, since otherwise $\mathtt{C}$ would not be a coalition with the largest cardinality among the coalitions whose joint utility is maximum in $\mathtt{IRCL}$.

Therefore, we can build a core stable and individually stable partition $\pi$ of $(N, U)$ iteratively in polynomial time as follows:

\medskip

\begin{compactitem}
\item[--] Form coalition $\mathtt{C} \in \mathtt{IRCL}$ with maximum joint utility, and break the ties in favor of coalitions with larger cardinality.
\item[--] Remove all coalitions that contain a member of $\mathtt{C}$ from $\mathtt{IRCL}$.
\item[--] Follow the same greedy rule until $\mathtt{IRCL}$ becomes empty.
\end{compactitem}

\medskip

Note that even if the above algorithm is exponential in the size of the number of agents $n$, it is polynomial in the size of $\mathtt{IRCL}$. \qed
\end{proof}

\section{Computing the Socially Optimal Solution}
\label{sec:ComOPT}

This section is devoted to establishing the computational complexity of finding a socially-optimal solution $\pi^*$ of a given HGCRP instance $\mathcal{G} = (N, U)$. The metric we use to evaluate the social welfare of a given solution $\pi$ is the utilitarian objective function \cite{agt}, i.e., the sum of the utilities of all agents. The social welfare $W(\pi)$ of a partition $\pi$ is then defined as $W(\pi) = \sum_{i \in N} u_i(\pi)$. A socially-optimal solution $\pi^*$ is a partition for which the social welfare is maximized.

Every socially-optimal solution is a perfect partition, provided it exists, since in a perfect partition all agents achieve their respective maximum attainable utilities. Notice that PERFECT-PARTITION polynomially reduces to the problem of finding a socially-optimal solution. This is because all it takes to decide whether a given socially-optimal solution is also a perfect partition is to verify that each agent is in her most preferred coalition, and that can trivially be done efficiently in IRCL representation. Hence, all the hardness results presented for PERFECT-PARTITION apply to the problem of computing the socially-optimal solution as stated by Corollary \ref{cor:FixedSocial}.

\begin{corollary} \label{cor:FixedSocial}
Finding a socially-optimal solution of a given HGCRP instance is {\bf NP-hard}, even under the restriction that the sizes of the coalitions are bounded above by $3$.
\end{corollary}

In this section, we improve upon this immediate result by proving that finding a socially-optimal solution of a given HGCRP instance is {\bf APX-complete}, even under the restriction that the coalition sizes are bounded above by $3$. For the general case, we prove that finding a socially-optimal solution is inapproximable within better than $O(n^{1-\epsilon})$ for any fixed $\epsilon > 0$, unless {\bf P = NP}, as stated by Theorem \ref{thm:Inapproximability}.

\begin{theorem} \label{thm:Inapproximability}
Finding a socially-optimal solution of a given HGCRP instance is inapproximable within a ratio better than $O(n^{1-\epsilon})$ for any fixed $\epsilon > 0$, unless {\bf P = NP}.
\end{theorem}

\begin{proof}
The proof is via an approximation preserving $A$-reduction \cite{apx-reduction} from the classical MAXIMUM-INDEPENDENT-SET problem. In this problem, we are given a graph $G = (V,E)$, and asked to find a maximum cardinality subset $I$ of vertices such that no pair of vertices in $I$ are adjacent. It is known that MAXIMUM-INDEPENDENT-SET is inapproximable within a ratio better than $O(n^{1-\epsilon})$ for any fixed $\epsilon > 0$, unless {\bf P = NP} \cite{n-inapproximable}.

\bigskip

For a given undirected graph $G=(V, E)$, we construct a corresponding HGCRP instance $(N,U)$ as follows:

\medskip

\begin{compactitem}
\item[-] For each edge $e \in E$, there is a corresponding agent in $N$,
\item[-] Each $v \in V$ corresponds to a coalition $C_v \subseteq N$, that consists of agents corresponding to the incident edges of $v$, with joint utility $U(C_v) = \frac{1}{\lvert C_v \rvert}$,
\item[-] For each subset $C \subseteq N$ of agents of $(N,U)$, for which there is no corresponding vertex in $G$, if $\lvert C \rvert = 1$ then $U(C) = \varepsilon$ where $0 < \varepsilon \leq \frac{1}{n^2}$, else $U(C) = 0$.
\end{compactitem}

\medskip

Notice that $(N,U)$ has positive joint utilities given only for the singleton coalitions, and the coalitions with a corresponding vertex in $G$. Hence, the size of $(N,U)$ in IRCL representation is polynomial in the size of $G$.

Let $\pi$ be a partition of $(N,U)$. Without loss of generality, assume that $\pi$ does not contain any coalition with $0$ joint utility. Since otherwise, we could replace each coalition $S$ with $0$ joint utility in $\pi$  with $|S|$ singleton coalitions with a total joint utility of $\varepsilon |S|$, and obtain another partition with higher social welfare.

Let $I$ denote the set of coalitions in $\pi$ with a corresponding vertex in $G$. For each distinct pair of coalitions $C_v \in I$ and $C_w \in I$, vertices $v$ and $w$ are nonadjacent in $G$, since otherwise the agent that corresponds to edge $(v,w)$ would reside in both $C_v$ and $C_w$. Therefore, $I$ corresponds to an independent set of $G$. Let $J$ denote the set of all other coalitions in $\pi$, i.e., $J = \pi \setminus I$. Notice that $J$ consists of singleton coalitions. Then, we have:

\bigskip

\begin{tabular}{ll}
$W(\pi)$ & $= \displaystyle\sum_{i \in N} u_i(\pi)$ \\
& \\
& $= \displaystyle\sum_{C \in I} U(C) \lvert C \rvert + \displaystyle\sum_{C \in J} U(C) \lvert C \rvert$ \\
& \\
& $= \displaystyle\sum_{C \in I} \frac{1}{\lvert C \rvert} \lvert C \rvert + \displaystyle\sum_{C \in J} \varepsilon$\\
& \\
& $= \lvert I \rvert + \varepsilon\lvert J \rvert$ \\
\end{tabular}

\bigskip

Suppose that $\pi$ is a $r$-approximation of a socially-optimal solution $\pi^*$ of $(N, U)$. That is, suppose that $\frac{W(\pi^*)}{W(\pi)} \leq r$ for some constant $r \geq 1$.
Similarly, we use $I^*$ to denote the set of coalitions formed in $\pi^*$ with a corresponding vertex in $G$, and we use $J^*$ to denote the set of coalitions formed in $\pi^*$ other than the ones in $I^*$. Note that $I^*$ corresponds to an independent set of $G$, $J^*$ consists of only singleton coalitions, and we have $W(\pi^*) = |I^*| + \varepsilon|J^*|$.

We now show that $I^*$ actually corresponds to a maximum independent set of $G$. For the sake of contradiction, assume that $I^*$ does not correspond to a maximum independent set of $G$. Then, there is an independent set $I'$ of $G$ such that $|I'| > |I^*|$. Let $J'$ be the set of uncovered edges by $I'$ in $G$, i.e., $J'$ is the set of edges that are not incident on any vertex in $I'$. Notice that $\pi' = I' \cup J'$ is a partition of $(N, U)$. Recall that $W(\pi') = |I'| + \varepsilon|J'|$.

To obtain contradiction, all we need is to show that $W(\pi') > W(\pi^*)$. First, notice that $\varepsilon \lvert J' \rvert - \varepsilon \lvert J^* \rvert \leq \frac{1}{n}$ due to the facts that $\lvert J' \rvert \leq n$, $\lvert J^* \rvert \leq n$ and $\varepsilon \leq \frac{1}{n^2}$. Moreover, notice that $\lvert I' \rvert - \lvert I^* \rvert \geq 1$ due to the facts that $\lvert I' \rvert > \lvert I^* \rvert$ and they are both integers. This means that $\lvert I' \rvert + \varepsilon \lvert J' \rvert > \lvert I^* \rvert + \varepsilon \lvert J^* \rvert$, and thus $W(\pi') > W(\pi^*)$. This contradicts with $\pi^*$ being a socially-optimal solution. Thus, $I^*$ corresponds to a maximum independent set of $G$.


The rest of the proof is as follows:

\bigskip

\begin{tabular}{ll}
$\displaystyle\frac{\lvert I^* \rvert}{ \lvert I \rvert}$ & $\leq$  $\displaystyle\frac{\lvert I^* \rvert}{\lvert I \rvert}
+ \displaystyle\frac{\varepsilon \lvert J^* \rvert}{\lvert I \rvert + \varepsilon \lvert J \rvert}$ \\
&\\
& $= \displaystyle\frac{\lvert I^* \rvert + \displaystyle\frac{\varepsilon \lvert J \rvert \lvert I^* \rvert}{\lvert I \rvert} + \varepsilon \lvert J^* \rvert}{\lvert I \rvert + \varepsilon\lvert J \rvert}$ \\
&\\
& $= \displaystyle\frac{\lvert I^* \rvert + \varepsilon \lvert J^* \rvert}{\lvert I \rvert + \varepsilon \lvert J \rvert} + \displaystyle\frac{\displaystyle\frac{\varepsilon \lvert J \rvert \lvert I^* \rvert}{\lvert I \rvert}}{\lvert I \rvert + \varepsilon \lvert J \rvert}$\\
&\\
& $= \displaystyle\frac{W(\pi^*)}{W(\pi)} + \displaystyle\frac{\varepsilon \lvert J \rvert \lvert I^* \rvert}{\lvert I \rvert^2 + \varepsilon \lvert I \rvert \lvert J \rvert}$ \\
\end{tabular}
\bigskip

Notice that we have $\varepsilon \lvert J \rvert \lvert I^* \rvert \leq 1$ since $\lvert J \rvert \leq n$, $|I^*| \leq n$, and $\varepsilon \leq \frac{1}{n^2}$. Furthermore, $\lvert I \rvert^2 \geq 1$ since $\lvert I \rvert \geq 1$. Therefore, we have $\frac{\varepsilon \lvert J \rvert \lvert I^* \rvert}{\lvert I \rvert^2 + \varepsilon\lvert I \rvert\lvert J \rvert} \leq 1$. This means that $\frac{\lvert I^* \rvert}{\lvert I \rvert} \leq r + 1$ since $\frac{W(\pi^*)}{W(\pi)} \leq r$, which finishes our $A$-reduction.

We showed that if we could found an $r$-approximate socially-optimal solution $\pi$ for $(N, U)$, then we could also found an $(r + 1)$-approximate maximum independent set $I$ for the given graph $G$. Since  MAXIMUM-INDEPENDENT-SET is inapproximable within better than $O(n^{1-\epsilon})$ for any fixed $\epsilon > 0$, unless {\bf P = NP}, so is finding a socially-optimal solution of a given HGCRP instance, unless {\bf P = NP}. \qed
\end{proof}

Since MAXIMUM-INDEPENDENT-SET problem is {\bf APX-complete} on cubic\footnote{Graphs in which all the vertices have degree $3$ are called cubic.} graphs \cite{apx-complete}, and $A$-reductions preserve membership in {\bf APX}, we also have the following corollary.

\begin{corollary} \label{cor:Inapproximability}
Finding a socially-optimal solution of a given HGCRP instance is {\bf APX-complete}, even when the sizes of the coalitions are bounded above by $3$.
\end{corollary}

\section{Quantification of Inefficiency}
\label{sec:QuaIne}

We now give tight bounds for the price of anarchy (PoA) and the price of stability (PoS) of HGCRP. The PoA and the PoS of HGCRP are defined as the supremum of the ratio of
the total utility of OPT to that of the socially worst and best core stable solutions, respectively, over all instances of the game.

\begin{theorem} \label{thm:Quality}
Both the price of anarchy and the price of stability of HGCRP are $n$, where $n$ is the number of agents.
\end{theorem}

\begin{proof}
For a given HGCRP instance, let $\pi$ and $\pi^*$ be a core stable solution, and a socially-optimal solution, respectively. For the sake of contradiction, assume that $W(\pi^*) > n W(\pi)$. But then, there exists an agent $i \in N$ such that $u_i(\pi^*) > W(\pi)$ by the pigeonhole principle, i.e., $U(\pi^*(i)) > W(\pi)$. This means that $\pi^*(i)$ is a blocking coalition with respect to $\pi$. This contradicts the fact that $\pi$ is a core stable partition. Therefore, $W(\pi^*) \leq n W(\pi)$ for any core stable partition $\pi$, specifically the socially worst one. Thus we have the following upper bound for the PoA:

$$ PoA \leq \frac{W(\pi^*)}{W(\pi)} \leq n$$

We now give a lower bound for the PoS of HGCRP by giving an example. Let $\mathcal{G} = (N, U)$ be an HGCRP instance, where $U$ is defined as $U(N) = 1$, $U(\{1\}) = 1 + \varepsilon$ for some $\varepsilon > 0$, and $U(G) = 0$ for all other subsets $G \subset N$ of agents. In a core stable partition $\pi$ of this game, $\{1\} \in \pi$, since otherwise $\{1\}$ is a blocking coalition with respect to $\pi$. Then, $N \notin \pi$ since $1 \in N$. Therefore, we have $W(\pi) = 1 + \varepsilon$. On the other hand, the socially-optimal solution $\pi^*$ is the one that $N$ is formed, i.e., $W(\pi^*) = n$. Therefore, we have the following lower bound for the price of stability:

$$\frac{W(\pi^*)}{W(\pi)} = \frac{n}{1 + \varepsilon} \leq PoS$$

Since the PoS is always less than or equal to the PoA, both of them are equal to $n$. \qed
\end{proof}

Since the PoA of HGCRP is $n$, by Theorem \ref{thm:Quality}, the social welfare of the solution returned by the greedy algorithm in Theorem \ref{thm:GreedyAlg} is within a factor of $n$ of that of the socially-optimal solution. Thus, it is actually an $n$-approximation algorithm for finding a socially-optimal solution, which proves Corollary \ref{cor:Approximability}. Notice that Corollary \ref{cor:Approximability} also reveals that the lower bound we give in Theorem \ref{thm:Inapproximability} is polynomially tight.

\begin{corollary} \label{cor:Approximability}
Finding a socially-optimal solution of a given HGCRP instance is $n$-approximable.
\end{corollary}

\section{Pairing Up vs. Staying Alone}
\label{sec:new}
In this section, we consider HGCRP with the additional restriction that the sizes of the coalitions are bounded above by $2$. This restriction coincides with the scenario where an agent may pair up with another agent or stay alone. This is especially a common situation in class project assignments where students may either work on a project with one of their classmates or may work by themselves, in which case the student is graded leniently for the sake of fairness. We show that all the computational problems we considered for HGCRP are tractable in this restricted setting. The project assignment example serves as an interesting application for the algorithms given in this section, which might be put into use to arrange project groups in classes so that the social welfare of the class (total of grades) is maximized, or efficiency is less sacrificed by assuring Pareto optimality of a stable solution.

The algorithms we give in this section make heavy use of the concept of matching in graph theory. A \textit{matching} $M$ is a set of non-adjacent edges, i.e., no two edges share a common vertex.  We say that a vertex is \textit{saturated} if it is an endpoint of one of the edges in the matching. A \textit{maximal matching} is a matching $M$ with the property that if any edge not in $M$ is added to $M$, it is no longer a matching. A \textit{maximum matching} is a matching $M$ such that the total weight of the edges is maximized. Notice that every maximum matching needs to be a maximal matching, but not vice versa.

\begin{theorem} \label{thm:MathchingAlg1}
A socially-optimal solution of a given HGCRP instance can be found in polynomial time if the sizes of the coalitions are bounded above by $2$.
\end{theorem}

\begin{proof}
Given an HGCRP instance $(N, U)$ where the sizes of the coalitions are bounded above by $2$, we construct an edge-weighted undirected graph $G = (V, E, w)$ with $w : E \rightarrow \mathbb{R}$ defined as follows:

\medskip

\begin{compactitem}
\item[--] For each agent $i \in N$, we add corresponding vertices $v_i \in V$ and $u_i \in V$.

\item[--] For each agent $i \in N$, we add an edge $e = (v_i, u_i)$, and we assign its weight as $w(e) = U(\{i\})$.

\item[--] For each unordered pair of agents $(i, j) \in N \times N$ where $i \neq j$, we add an edge $e = (v_i, v_j)$, and we assign its weight as $w(e) = 2 \times U(\{i, j\})$.
\end{compactitem}

\medskip

Notice that each edge of $G$ corresponds to a coalition in our construction. Namely, an edge $(v_i, u_i)$ corresponds to the singleton coalition $\{i\}$, and an edge $(v_i, v_j)$ corresponds to the coalition $\{i, j\}$. Note that since the sizes of the coalitions are bounded above by $2$ in $(N, U)$, each possible coalition is represented by an edge of $G$. Furthermore, the weight of an edge $e$ is equal to the total utility of the members of the corresponding coalition, since we have $w(e) = U(\{i\})$ if $e$ corresponds to the singleton coalition $\{i\}$, and we have $w(e) = 2 \times U(\{i, j\})$ if $e$ corresponds to the coalition $\{i, j\}$ whose size is $2$.

In a maximal matching $M$ of $G$, if $u_i$ is not saturated, then $v_i$ is saturated since $u_i$ is only adjacent to $v_i$. Since otherwise $M \cup (v_i, u_i)$ would be a matching, which would contradict that $M$ is a maximal matching. Thus, at least one of $u_i$ or $v_i$ is saturated in a maximal matching $M$. If both $u_i$ and $v_i$ are saturated in a matching $M$ then $(v_i, u_i) \in M$, and consequently $(v_i, v_j) \notin M$ for any agent $j \in N$. Therefore, a maximal matching $M$ of $G$ corresponds to a partition of the set of agents $N$. Thus, a maximum matching of $G$ corresponds to a socially-optimal partition of the set of agents $N$.

Since a maximum matching of a graph can be found in polynomial time by the Blossom algorithm \cite{edmonds}, so does a socially-optimal solution of a given HGCRP instance, if the sizes of the coalitions are bounded above by $2$. \qed
\end{proof}

Since any socially-optimal partition is both a Pareto optimal partition and a perfect partition (if exists), Theorem \ref{thm:MathchingAlg1} has the following corollary:

\begin{corollary}
Both of the following computational problems are tractable for a given HGCRP instance, if the sizes of the coalitions are bounded above by $2$.
    \begin{itemize}
        \item Finding a perfect partition (if exists),
        \item Finding a Pareto optimal partition.
    \end{itemize}
\end{corollary}

We next prove that the problem of finding a Pareto optimal, core stable and individually stable partition is also tractable in this restricted setting.

\begin{theorem}
\label{thm:MathchingAlg2}
A Pareto optimal, core stable, and individually stable partition of a given HGCRP instance can be found in polynomial time, if the sizes of the coalitions are bounded above by $2$.
\end{theorem}

\begin{proof}
Suppose that we are given an HGCRP instance $(N, U)$ where the sizes of the coalitions are bounded above by $2$. Let $\mathcal{C}$ denote the set of coalitions, i.e., the set of all non-empty subsets of agents $N' \subseteq N$ such that $|N'| \leq 2$. Let $\mathcal{B}$ denote the set of all coalitions which have the maximum joint utility $\beta = max_{C \in \mathcal{C}} U(C)$, i.e., $\mathcal{B} = arg\,max_{C \in \mathcal{C}} U(C)$.

\bigskip

We construct an edge-weighted undirected graph $G = (V, E, w)$ with $w : E \rightarrow \mathbb{R}$ defined as follows:

\medskip

\begin{compactitem}
\item[--] For each agent $i \in N$, we add corresponding vertices $v_i \in V$ and $u_i \in V$.

\item[--] For each coalition $B \in \mathcal{B}$, we add a corresponding edge as follows:

\begin{compactitem}
\item[-] If $B = \{i\}$ for some agent $i \in N$, then we add the edge $e =(v_i, u_i)$ with $w(e) = 1$.

\item[-] If $B = \{i, j\}$ for some pair of agents, then we add the edge $e =(v_i, v_j)$ with $w(e) = 2$.
\end{compactitem}

\end{compactitem}

\medskip

We then find a maximum matching $M$ of $G$ in polynomial time by using the Blossom algorithm \cite{edmonds}. Let $\mathcal{M}$ denote the set of coalitions that corresponds to the edges in  $M$. None of the coalitions in $\mathcal{M}$ have common members, since no two edges in $M$ share a common vertex.

Suppose that $\pi$ is a partition where all of the coalitions in $\mathcal{M}$ are formed. Notice that partition $\pi$ maximizes the number of agents with utility $\beta$. This is because $M$ is a maximum matching and the weight of each edge in $G$ is equal to the size of its corresponding coalition. We divide the rest of the proof into three parts.

\paragraph{Pareto Optimality:} Suppose that there exists a partition $\pi'$ that Pareto dominates $\pi$. Since $\beta$ is the maximum utility an agent can attain, all the agents that are in a coalition in $\mathcal{M}$ must attain a utility of $\beta$ in $\pi'$ as well. Since the number of agents that are in a coalition in $\mathcal{M}$ is the maximum number of agents that can attain a utility of $\beta$ in any partition, no other agent can attain a utility of $\beta$ in $\pi'$. Therefore, without loss of generality, we can assume that all of the coalitions in $\mathcal{M}$ are still formed in $\pi'$.

\paragraph{Core Stability:} Since the joint utility of any coalition $C \in \mathcal{M}$ is greater than or equal to the joint utility of any possible coalition, no member of $C$ can participate in a blocking coalition with respect to $\pi$. This means that if there exists a blocking coalition $S$ with respect to $\pi$ then all of the coalitions in $\mathcal{M}$ are still formed in $\pi_S$.

\paragraph{Individual Stability:} Since the joint utility of any coalition $C \in \mathcal{M}$ is greater than or equal to the joint utility of any possible coalition, no member of $C$ can move to another coalition and get better off. Moreover, if $\lvert C \rvert = 2$ then surely no agent can move to coalition $C$ because of the restriction on the sizes of the coalitions. Consider a singleton coalition $C = \{i\} \in \mathcal{M}$. It may be the case that an agent $j$ moves into $C$ and gets better off. However, this would make agent $i$ worse off, since otherwise the number of agents that attains a utility of $\beta$ would have increased, which contradicts $M$ being a maximum matching of $G$.

Thus, we showed that once the coalitions in $\mathcal{M}$ are formed, we do not need to worry about them anymore as we construct a Pareto optimal, core stable, and individually stable partition. Therefore, we can remove the agents that are in a coalition in $\mathcal{M}$, and repeat the same procedure for the remaining agents until $\mathcal{C} = \emptyset$. \qed
\end{proof}

\section{Conclusion}
\label{sec:conclusion}

We provided a comprehensive study of hedonic games possessing common ranking property, which is a natural model for the formation of general partnerships \cite{law}. We strengthened the landmark result that every HGCRP instance has a core stable coalition structure, by proving that every HGCRP instance has a Pareto optimal, core stable and individually stable coalition structure. The economic significance of our result is that efficiency is not to be totally sacrificed for the sake of stability in HGCRP.

We established the computational complexity of several problems related to HGCRP both for the general case, and for the special cases, where the sizes of the coalitions are bounded above. Our investigations revealed that all the computational problems  (except for finding a core stable and individually stable partition) we considered are intractable in the general case. For the special case, we discovered a distinctive line that makes all of these problems intractable. Namely, all of them are:
\medskip
\begin{compactitem}
\item[--] intractable when the sizes of the coalitions are bounded above by $3$,
\item[--] and tractable when the sizes of the coalitions are bounded above by $2$.
\end{compactitem}
\medskip

We also provided an inapproximability lower bound for the problem of finding a socially-optimal solution, which we proved to be polynomially tight. Lastly, we quantified the loss of efficiency in HGCRP due to the selfish behavior of agents by proving tight bounds on the price of anarchy, and the price of stability.

\begin{acknowledgements}
We would like to thank Edith Elkind for extending our existence result (Theorem \ref{thm:Existence}) so that it includes individual stability (Lemma \ref{lem:IS}), and also for pointing out the counterexample for Nash stability (Example \ref{Ex3}).
\end{acknowledgements}

\noindent
\textbf{Funding} This work is supported by The Scientific and Technological Research Council of Turkey (TÜBİTAK) through grant 118E126.

\medskip

\noindent
\textbf{Financial/Non-financial Interests} The authors have no relevant financial or non-financial interests to disclose.

\bibliographystyle{spmpsci}      


\end{document}